\documentclass{article}

\usepackage{PRIMEarxiv}

\usepackage[utf8]{inputenc} 
\usepackage[T1]{fontenc}    
\usepackage{hyperref}       
\usepackage{url}            
\usepackage{booktabs}       
\usepackage{amsfonts}       
\usepackage{nicefrac}       
\usepackage{microtype}      
\usepackage{lipsum}
\usepackage{fancyhdr}       
\usepackage{graphicx}       
\usepackage{anyfontsize}
\usepackage{multirow}%
\usepackage{amsmath,amsfonts}%
\usepackage{amsthm}%
\usepackage{mathrsfs}%
\usepackage[title]{appendix}%
\usepackage{xcolor}%
\usepackage{textcomp}%
\usepackage{manyfoot}%
\usepackage{booktabs}%
\usepackage{algorithm}%
\usepackage{algorithmicx}%
\usepackage{algpseudocode}%
\usepackage{listings}%
\usepackage[english]{babel}
\usepackage{mathtools}
\usepackage{booktabs}
\usepackage{float}
\usepackage{wrapfig}
\usepackage{comment}
\usepackage{subcaption}

\newcommand{\defeq}{\vcentcolon=}

\newtheorem{theorem}{{\bf Theorem}}
\newtheorem{lemma}{{\bf Lemma}}

\newtheorem{remark}{{\bf Remark}}

\title{Accelerating Spectral Clustering on Quantum and Analog Platforms}

\author{
  Xingzi Xu \thanks{Work performed as an intern in the Applied Sciences Lab at SRI International.} \\
  Department of Electrical and Computer Engineering\\
  Duke University \\
  Durham, NC, USA\\
  \texttt{xingzi.xu@duke.edu} \\
   \And
  Tuhin Sahai \\
  SRI International \\
  Menlo Park, CA, USA\\
  \texttt{tuhin.sahai@sri.com} \\
}

\begin{document}
\maketitle

\begin{abstract}
We introduce a novel hybrid quantum-analog algorithm to perform graph clustering that exploits connections between the evolution of dynamical systems on graphs and the underlying graph spectra. This approach constitutes a new class of algorithms that combine emerging quantum and analog platforms to accelerate computations. Our hybrid algorithm is equivalent to spectral clustering and significantly reduces the computational complexity from $\mathcal{O}(N^3)$ to $\mathcal{O}(N)$, where $N$ is the number of nodes in the graph. We achieve this speedup by circumventing the need for explicit eigendecomposition of the normalized graph Laplacian matrix, which dominates the classical complexity, and instead leveraging quantum evolution of the Schr\"{o}dinger equation followed by efficient analog computation for the dynamic mode decomposition (DMD) step. Specifically, while classical spectral clustering requires $\mathcal{O}(N^3)$ operations to perform eigendecomposition, our method exploits the natural quantum evolution of states according to the graph Laplacian Hamiltonian in linear time, combined with the linear scaling for DMD that leverages efficient matrix-vector multiplications on analog hardware. We prove and demonstrate that this hybrid approach can extract the eigenvalues and scaled eigenvectors of the normalized graph Laplacian by evolving Schr\"{o}dinger dynamics on quantum computers followed by DMD computations on analog devices, providing a significant computational advantage for large-scale graph clustering problems. Our demonstrations can be reproduced using our code that has been released at  \url{https://github.com/XingziXu/quantum-analog-clustering.}
\end{abstract}

\section{Introduction}\label{sec:intro}

Graph clustering is a popular technique to identify and group densely connected subgraphs within a larger graph. It is a powerful decomposition approach that enables the analysis of interconnected systems for a variety of applications in a wide range of fields, such as social networks \cite{mishra2007clustering,liu2014weighted,bu2018gleam}, fraud detection \cite{sabau2012survey,pourhabibi2020fraud}, bioinformatics \cite{higham2007spectral}, uncertainty analysis in networked systems~\cite{surana2012iterative}, decomposition for scientific computation~\cite{klus2011efficient}, and transport networks \cite{ahmed2007cluster}. Consequently, the versatility of graph clustering makes it an invaluable tool in areas ranging from scientific discovery to industrial applications. Spectral clustering, a prevalent graph decomposition approach, is particularly advantageous in high-dimensional spaces where the geometric characteristics of the data are not readily apparent \cite{von2007tutorial}. Unlike traditional clustering methods, spectral clustering uses the properties of the eigenvalues and eigenvectors of the Laplacian of the underlying graph (or data) to find the optimal partitioning. This aspect enables spectral clustering to cluster points based on their interconnectedness rather than their raw distances from each other, typically resulting in better performance when dealing with data where the clusters are irregular, intertwined, or lie on a complex surface.

In this work, we propose a highly scalable quantum-analog hybrid algorithm based on the evolution of Schr\"{o}dinger dynamics and dynamic mode decomposition (DMD) to rapidly and accurately retrieve spectral clustering information. We evolve the underlying Schr\"{o}dinger dynamics on quantum devices~\cite{sahai2023spectral}, and perform the subsequent spectrum computations on analog machines. The approach has a scaling of $\mathcal{O}(N)$, where $N$ is the number of nodes in the graph. Since existing state-of-the-art classical methods scale as $\mathcal{O}(N^3)$, our approach provides a polynomial speed-up over traditional computing platforms. To our best knowledge, this is the first-of-a-kind algorithm that exploits the combination of quantum and analog computing platforms to extract a scaling superior to algorithms that run on a single computing platform.

In quantum computing, qubits, the fundamental elements of quantum information, evolve according to the time-dependent Schr\"{o}dinger's equation. Qubits differ from classical bits in that they can exist in a superposition of states, meaning that a qubit can be in states representing the traditional bits (0, 1) or a superposition of both. Quantum algorithms provide substantial reductions in computational complexity for specific problems such as integer factorization~\cite{shor1999polynomial}, database search~\cite{grover1996fast}, solving linear systems~\cite{harrow2009quantum}, simulating dynamical systems~\cite{surana2024efficient}, and matrix multiplication~\cite{li2021quantum} to name a few.
On the other hand, analog computing with resistive memory has a fast response speed and embeds parallelism, significantly accelerating matrix computation \cite{sun2022invited}, a core step in a wide variety of numerical techniques. Analog computing has garnered attention due to its potential to reduce the data movement bottleneck in traditional computing architectures. The essence of analog computing is performing computational operations directly at the data's location rather than moving data between memory and the central processing unit (CPU). Integrating crosspoint arrays with negative feedback amplifiers allows for addressing linear algebra challenges like solving linear systems and computing eigenvectors in a single step. Analog computing improves the exponential solving of linear systems and has a computational complexity of $\mathcal{O}(1)$ for matrix-vector multiplication \cite{sun2022invited,sun2020time}. 

Despite the advantages of spectral clustering, it typically scales as $\mathcal{O}(N^3)$ (where $N$ is the number of nodes in the graph), limiting its applications in large graphs \cite{spectralScott}. We show that spectral clustering has an $\mathcal{O}(N)$ complexity on a hybrid quantum-analog platform. 

We organize the paper as follows: Section \ref{sec:spectral} provides a brief overview of classical spectral clustering. Section \ref{sec:dmd} delves into the classical decentralized algorithm for spectral clustering that leverages the wave equation and Dynamic Mode Decomposition (DMD). In sections~\ref{sec:wave} and~\ref{sec:eigen}, we discuss our novel approach that uses quantum evolution of wave dynamics and singular value- and eigen- decompositions on analog computers. We also provide details on efficiently solving the linear system of equations for computing eigenvectors on quantum or analog computers in section~\ref{sec:linear}. We present experimental results in section \ref{sec:result} and wrap up with conclusions in section \ref{sec:conclusion}.

\section{Analog matrix computing (AMC) circuits}
\label{sec:analog}
Conventional computers have significantly improved in speed and efficiency over the past decades, as captured by Moore's law, stating that the number of transistors in an integrated circuit doubles about every two years\cite{moore}. However, the industry is approaching a limit where it is increasingly difficult to continue with appreciable reductions in the size of transistors~\cite{moore_quantum}. Additionally, due to the high computational complexity of matrix-vector multiplication required for the training and inference of artificial intelligence models on digital platforms, there is an urgent need to develop alternative computing devices with more attractive scaling properties for these computations. Due to the inherent architectural parallelism for matrix operations, analog matrix computing (AMC) circuits based on crosspoint resistive memory arrays or photonics provide a desirable alternative \cite{sun2022invited}. We note that, historically, analog computers have been used extensively during the Apollo space program to simulate the trajectory of spacecraft. They fell out of vogue during the digital revolution, but have recently been becoming increasingly popular for a wide variety of applications.


\begin{figure}[hbt!]
  \centering
  \includegraphics[width=0.4\textwidth, trim=20pt 20pt 100pt 15pt]{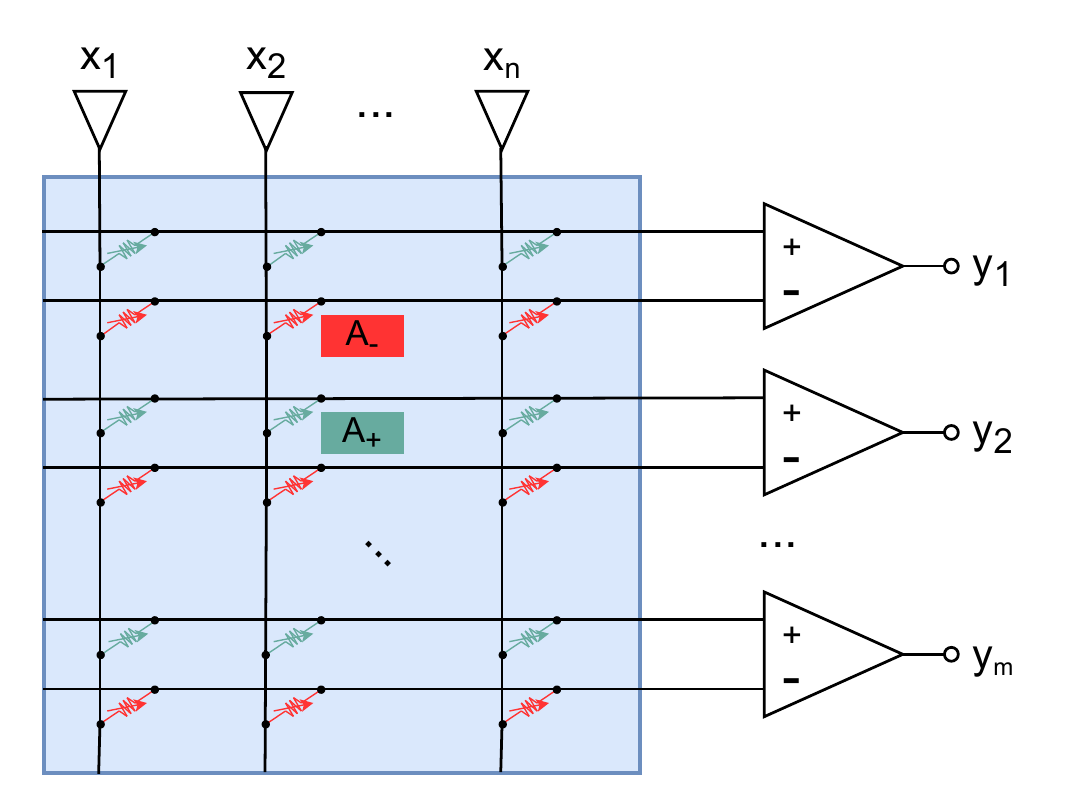}
  \caption{Schematic of AMC circuits for matrix-vector multiplication between matrix $\mathbf{A}$ and vector $\mathbf{x}$, where $\mathbf{A}=\mathbf{A}_+-\mathbf{A}_-$\cite{sun2022invited}. $\mathbf{A}\in\mathcal{R}^{M\times N}$, $\mathbf{x}\in\mathcal{R}^{N\times 1}$, $\mathbf{y}\in\mathcal{R}^{M\times 1}$.}
  \label{fig:amc}
\end{figure}

For example, the multiplication between an $M\times N$ matrix and an $N\times 1$ vector, has the computational complexity of $\mathcal{O}(MN)$ on digital computers and $\mathcal{O}(1)$ on AMC circuits~\cite{sun2022invited} or photonic devices~\cite{xu2022high}. Figure \ref{fig:amc} shows the schematic of a prototypical AMC circuit for computing the matrix-vector multiplication (MVM) between a matrix $\mathbf{A}$ and a vector $\mathbf{x}$. The conductance in the circuits captures the entries of $\mathbf{A}$, and the applied voltages encode the values of $\mathbf{x}$. The output currents provide the result of $\mathbf{y}=\mathbf{Ax}$, due to Ohm's law \cite{ohm}. Since conductance can only take positive values, matrices with negative entries can be split as $\mathbf{A}=\mathbf{A}_+-\mathbf{A}_-$, where $\mathbf{A}_+$ is $\mathbf{A}_-$ are both positive. The results are then calculated as $\mathbf{A}\mathbf{x}=\mathbf{A}_+\mathbf{x}-\mathbf{A}_-\mathbf{x}$. Unlike digital computers, analog computers compute the MVM results for all entries in parallel after the setting of conductances and application of voltages.
To quantify the error and computational time of MVM on analog computers, we prove the following theorem using circuit dynamics.


\begin{theorem}
Let the minimum time for convergence up to an of error $\epsilon$ (where $\epsilon$ is the difference between the computed and ground truth solutions) of matrix-vector multiplication between an $N\times N$ matrix $\mathbf{A}$ and an $N\times 1$ vector $\mathbf{x}$ on circuit-based analog computers be $T$. Then $T$ is given by,
\begin{align}
    T = \displaystyle\max_{j}\frac{1}{\beta_j}\ln(\frac{\mathbf{x}_{0}\beta - \alpha_j}{\epsilon}),
\end{align}
\label{thm:analog_stb}
where $\mathbf{x}_0$ is the initial condition, $\alpha_j=\vartheta\omega\frac{\sum_{i}\mathbf{A}_{ji}\mathbf{y}_{i}}{1+\sum_{i}\mathbf{A}_{ji}}$, $\beta=\frac{\vartheta\omega}{1+\sum_{i}\mathbf{A}_{ji}}$, and $j$ is the index that takes values $\{1,2,\hdots,N\}$. 
\end{theorem}
\begin{proof}
Given the evolution dynamics of analog circuits is given by~\cite{sun2021circuit},
\begin{align}
    \frac{d\mathbf{x}(t)}{dt} = \vartheta\omega \mathcal{U}(\mathbf{A}\mathbf{y} - \mathbf{x}(t)),
\label{eq:analog_dyn}
\end{align}
where $\vartheta$ is the DC loop gain, $\omega$ is the 3-dB bandwidth, $\mathbf{y}$ is an input vector, $\mathbf{A}$ is the conductance matrix, and 
\begin{align}
  \mathcal{U} =  \begin{pmatrix}
\frac{1}{1+\sum_{i}\mathbf{A}_{1i}} & 0   & \cdots & 0   \\
0   & \frac{1}{1+\sum_{i}\mathbf{A}_{2i}} & \cdots & 0   \\
\vdots & \vdots & \ddots & \vdots \\
0   & 0   & \cdots & \frac{1}{1+\sum_{i}\mathbf{A}_{ni}} \\
\end{pmatrix}.
\end{align}
By expanding the above equations, it is easy to show that the $j$-th equation is given by,
\begin{align}
     \frac{d\mathbf{x}_j(t)}{dt} = \alpha_j - \beta_{j}\mathbf{x}_j(t),
\label{eq:jeq}
\end{align}
where $\alpha_j$ and $\beta$ are constants given by,

\begin{align}
    \alpha_j&=\vartheta\omega\frac{\sum_{i}\mathbf{A}_{ji}\mathbf{y}_{i}}{1+\sum_{i}\mathbf{A}_{ji}},\nonumber\\
    \beta_j &= \frac{\vartheta\omega}{1+\sum_{i}\mathbf{A}_{ji}}.
\end{align}
Note that Eqns.~\ref{eq:jeq} are all decoupled and their solution at time $t$ is given by,
\begin{align}
 \mathbf{x}(t) = \frac{(\mathbf{x}(0)\beta_{j} - \alpha_{j})e^{-\beta_{j}t} + \alpha_{j}}{\beta_j}. 
\end{align}

Now, let the solution as $t\rightarrow\infty$ be $\mathbf{x}_f = \frac{\alpha_j}{\beta_j}$. If we consider a time $T_j$
 such that the error $||\mathbf{x}(t)-\mathbf{x}_f||$ is less than equal to our threshold $\epsilon$ we get,
 
\begin{align}
    \frac{(\mathbf{x}(0)\beta_{j} - \alpha_{j})e^{-\beta_{j}t}}{\beta_j} \leq \epsilon,
\end{align}
or,
\begin{align}
  T_j \geq \frac{1}{\beta_j}\ln(\frac{\mathbf{x}_{0}\beta - \alpha_j}{\epsilon}).
\end{align}
It now follows that the minimum time for convergence of the system (as indicated by the error reducing below a threshold of $\epsilon$) is given by maximum over all $j\in\{1,2,\hdots,N\}$,
\begin{align}
    T = \displaystyle\max_{j}\frac{1}{\beta_j}\ln(\frac{\mathbf{x}_{0}\beta - \alpha_j}{\epsilon}).
\end{align}
\end{proof}
We note that the above derivation not only provides performance bounds on the time to convergence, but also enables the design of efficient analog circuits. Moreover, the above derivation will provide the wall clock time for  executing various steps within our algorithm (listed in table~\ref{table:steps}) on analog platforms. 

\section{Preliminaries}
In this section, we will summarize our previously constructed classical spectral clustering algorithm that exploits wave dynamics. This algorithm will subsequently be adopted for quantum and analog platforms. Before we get into the details spectral clustering and our approach, we will provide a summary of existing methods for clustering.
\subsection{Overview of various methods for graph clustering}
The overall goal of partitioning graphs is to separate nodes into groups or clusters such that there are ``weak'' interactions between clusters and strong interactions within the groups. Numerous methods for partitioning graphs have been developed over the years. A summary of some of the approaches and their drawbacks can be found in table~\ref{tab:graph_clustering}. Several of the methods require additional information about the graph. We note that algorithms for modularity maximization are popular for clustering large graphs and have been used extensively for community detection. As shown in the table, the Louvain method has a scaling of $\mathcal{O}(N \log N)$, which is significantly better than the $\mathcal{O}(N^3)$ scaling of spectral clustering. Note that modularity maximization has a spectral formulation~\cite{newman2013spectral}. Consequently, our hybrid analog-quantum algorithm can easily adapt to the community detection (modularity maximization) scenario. We can achieve this by replacing the graph Laplacian in our approach with the matrix that arises in the modularity maximization formulation.

\begin{table}
\centering
\caption{Various Graph Clustering Methods}
\label{tab:graph_clustering}
\begin{tabular}{llll}
\toprule
                                   Method &                                  Approach &                                  Complexity &                        Drawbacks \\
\midrule
                      Spectral Clustering &          Eigen-decomposition &                  $\mathcal{O}(N^3)$  & Poor scalability \\
Louvain method &         Modularity maximization &                                 $\mathcal{O}(N \log N)$ &                Requires resolution parameter\\
                  Label Propagation  &                 Label diffusion &                                        $\mathcal{O}(E)$ &  Needs label data \\
                  Markov Clustering (MCL) &         Random walks  &         $\mathcal{O}(N^3)$  &                            Unbalanced cuts \\
                  Hierarchical Clustering &      Agglomerative or divisive &              $\mathcal{O}(N^2)$  &               Assumes hierarchical structure \\
\bottomrule
\end{tabular}
\end{table}

\subsection{Short introduction spectral clustering}
\label{sec:spectral}
Consider a graph denoted by $\mathcal{G}=(V,E)$ that consists of a node set $V={1,\dots,N}$ and an edge set $E\subseteq V\times V$. Each edge $(i,j)\in E$ in this graph is associated with a weight $\mathbf{W}_{ij}>0$, and $\mathbf{W}$ represents the $N\times N$ weighted adjacency matrix of $\mathcal{G}$. In this context, $\mathbf{W}_{ij}=0$ only if $(i,j)\not\in E$.
The normalized graph Laplacian, represented by $\mathbf{L}$, is defined as follows:
\begin{align}
\label{eqn:glap}
\mathbf{L}_{ij}=\begin{cases}
    1, &\text{if } i=j,\\
    -\mathbf{W_{ij}}/\sum_{l=1}^N\mathbf{W}_{il}, &\text{if } (i,j)\in E,\\
    0, &\text{otherwise,}
\end{cases}\end{align}
An equivalent representation of the normalized graph Laplacian is $\mathbf{L}=\mathbf{I}-\mathbf{D}^{-1}\mathbf{W}$, where $\mathbf{D}$ is a diagonal matrix that consists of the sums of the rows of $\mathbf{W}$. We focus solely on undirected graphs, for which the Laplacian is symmetric, and the eigenvalues are real numbers. The eigenvalues of $\mathbf{L}$ can be arranged in ascending order as $0=\lambda_1\leq\lambda_2\leq\dots\leq\lambda_N$. Each eigenvalue corresponds to an eigenvector, denoted as $\mathbf{v}^{(1)},\mathbf{v}^{(2)},\dots,\mathbf{v}^{(N)}$, where $\mathbf{v}^{(1)}=\mathbf{1}=[1,1,\dots,1]^T$ \cite{von2007tutorial}. For our study, we assume $\lambda_1<\lambda_2$, which means that the graph does not contain disconnected clusters. Spectral clustering partitions the graph $\mathcal{G}$ into two clusters by employing the signs of the entries of the second eigenvector $\mathbf{v}^{(2)}$. $k-$means clustering on the signs of entries of higher eigenvectors enables one to partition graphs into more than two clusters \cite{von2007tutorial}. The spectral gap in the eigenvalues of the graph Laplacians typically determines the number of clusters.

\subsection{Wave equation-based classical clustering method}
\label{sec:dmd}
As in \cite{sahai2012hearing,zhu2022dynamic}, we consider the wave equation given by,
\begin{align}
\label{eqn:wave}
    \frac{\partial^2 \mathbf{u}}{\partial t^2} = c^2 \Delta \mathbf{u}.
\end{align}
The discrete form of the wave equation on a graph is given by,
\begin{align}
\label{eqn:wave_discrete}
    \mathbf{u}_i(t) = 2\mathbf{u}_i(t-1)-\mathbf{u}_i(t-2)-c^2\sum_{j\in\mathcal{N}(i)} \mathbf{L}_{ij}\mathbf{u}_j(t-1),
\end{align}
where $\mathcal{N}(i)$ is the set of neighbors of node $i$ including the node $i$ itself \cite{Friedman2004WaveEF}. To update $\mathbf{u}_i$, one only needs the previous value of $\mathbf{u}_j$ at the neighboring nodes and the connecting edge weights. 

Given the initial condition $\frac{\mathrm{d}\mathbf{u}}{\mathrm{d}t}|_{t=0}=0$ and $0<c<\sqrt{2}$, the solution of the wave equation~\ref{eqn:wave_discrete} can be written as
\begin{align}
\label{eqn:wave_dyn}
    \mathbf{u}(t) = \sum_{j=1}^N \mathbf{u}(0)^T\mathbf{v}^{(j)}(p_j e^{it\zeta_j}+q_je^{-it\zeta_j})\mathbf{v}^{(j)},
\end{align}
where $p_j=(1+i\tan(\zeta_j/2))/2$, $q_j=(1-i\tan(\zeta_j/2))/2$ \cite{zhu2022dynamic}. Consequently, computing the eigenvectors of $\mathbf{L}$ is transformed into the computation of the coefficients of the frequencies of the wave dynamics in equation~\ref{eqn:wave_dyn}. 

To extract the above coefficients, we use Dynamic Mode Decomposition (DMD), a powerful tool for analyzing the dynamics of nonlinear systems. We interpret DMD as an approximation of Koopman modes~\cite{2013OnDM}. We apply DMD to one-dimensional time series data with time-delay embedding by constructing the matrix $\mathbf{X}$ and $\mathbf{Y}$ for the exact DMD matrix $\mathcal{A}=\mathbf{Y}\mathbf{X}^+$ (where $^+$ denotes the pseudoinverse) \cite{dylewsky2022principal}. When DMD with time delay embedding is implemented on one-dimensional signals of the form,
\begin{align} 
\label{eqn:1dsig}
    u(t)=\sum_{j=1}^J a_j e^{i\zeta_j t}, 
\end{align} 
where $\zeta_j\in(-\pi,\pi)$, $j=1,2,\dots,J$ are unique frequencies, it can successfully extract the coefficients $a_j$ \cite{zhu2022dynamic,dylewsky2022principal}. 

We form the following matrices with the one-dimensional signal $u(t)$,
\begin{align}
\label{eqn:X}
    \mathbf{X}&\defeq\begin{bmatrix}
    u(0) & u(1) & \dots & u(M-1)\\
    u(1) & u(2) & \dots & u(M)\\
    \vdots & \vdots & \ddots & \vdots\\
    u(K-1) & u(K) & \dots & u(K+M-2)\\
\end{bmatrix},\\
&= \begin{bmatrix}
    \mathbf{x}^l(0) & \mathbf{x}^l(1) & \dots & \mathbf{x}^l(M-1)
\end{bmatrix},\\
\label{eqn:Y}
\mathbf{Y}&\defeq\begin{bmatrix}
    \mathbf{x}^l(1) & \dots & \mathbf{x}^l(M-1) & \mathbf{x}^l(M)
\end{bmatrix},
\end{align}
where $\mathbf{x}(t)\defeq[u(t),u(t+1),\dots,u(t+(M-1))]^T$. Let $H>1$, also define
\begin{align}
\label{eqn:Phi}
    \Phi_H &\defeq \begin{bmatrix}
        1 & 1 & \dots & 1\\
        e^{i\zeta_1} & e^{i\zeta_2} & \dots & e^{i\zeta_N}\\
        \vdots & \vdots & \ddots & \vdots\\
        e^{i(H-1)\zeta_1} & e^{i(H-1)\zeta_2} & \dots & e^{i(H-1)\zeta_N}\\
    \end{bmatrix},\\
    &=[\phi_1,\phi_2,\dots,\phi_N],
\end{align}
where $\phi_j=[1,e^{i\zeta_j},\dots,e^{i(H-1)\zeta_j}]^T$. \cite{zhu2022dynamic} proved the following two lemmas to connect the above matrices and the associated DMD computations to the spectra of the graph Laplacian.
\begin{lemma}
\label{lemma:dmd}
    For one-dimensional signal $u(t)$ defined by equation \eqref{eqn:1dsig}, if $K\geq J$ and $M\geq J$ of the matrices $\mathbf{X}$ and $\mathbf{Y}$ defined by \eqref{eqn:X} and \eqref{eqn:Y} respectively, the eigenvalues of $\mathcal{A}=\mathbf{Y}\mathbf{X}^+$ are $\{e^{i\zeta_j}\}_{j=1}^J$ and the columns of $\Phi_K$ defined by \eqref{eqn:Phi} are the corresponding eigenvectors \cite{zhu2022dynamic}.
\end{lemma}

\begin{lemma}
\label{lemma:wave}
    At any node $l$, the DMD computations on matrices $\mathbf{X}(u_l)$, $\mathbf{Y}(u_l)$ using local snapshots $\mathbf{u}_l(0), \mathbf{u}_l(1),\dots, \mathbf{u}_l(4N-1)$ defined by Eqn.~\ref{eqn:wave_dyn} with $K=M=2N$ yields exact eigenvalues of the Laplacian and the corresponding eigenvectors (scaled) \cite{zhu2022dynamic}.
\end{lemma}

Based on lemma~\ref{lemma:wave}, \cite{zhu2022dynamic} developed a DMD-based algorithm for distributed spectral clustering, given in algorithms~\ref{alg:dmd}, \ref{alg:wave_cluster}.

\begin{algorithm}
\caption{DMD$(\mathbf{X}, \mathbf{Y})$: For computing eigenvalues and eigenvector components at node $i$ \cite{zhu2022dynamic}.}\label{alg:dmd}
\begin{algorithmic}[1]
   \State \text{Compute reduced SVD of} $\mathbf{X}$\text{, i.e., }$\mathbf{X}=\mathbf{U}\mathbf{\Sigma}\mathbf{V}^*$.
   \State \text{Define the matrix } $\mathcal{\Tilde{A}}\equiv \mathbf{U}^*\mathbf{Y}\mathbf{V}\mathbf{\Sigma}^{-1}$
   \State \text{Compute eigenvalues/vectors $\mu$ and $\xi$ of $\mathcal{\Tilde{A}}$, i.e., $\mathcal{\Tilde{A}}\xi=\mu\xi$. Nonzero eigenvalues $\mu$ are DMD eigenvalues.}
   \State \text{The DMD mode corresponding to $\mu$ is then given by $\hat{\phi}\equiv \frac{1}{\mu}\mathbf{YV\Sigma}^{-1}\xi$.}
   \State \parbox[t]{\linewidth}{%
    Compute $\hat{\mathbf{a}}$ by solving the linear system $\hat{\mathbf{\Phi}}\hat{\mathbf{a}}=\mathbf{x}(0)$, where the columns of $\hat{\mathbf{\Phi}}$ are the eigenvectors sorted in decreasing order based on the real part of the eigenvalues.}
   \State $a_i^{(j)}=\hat{\phi}_{j,1}\hat{a}_j, j=1,2,\dots,k.$
\end{algorithmic}
\end{algorithm}

\begin{algorithm}
\caption{Wave equation based graph clustering \cite{zhu2022dynamic}.}\label{alg:wave_cluster}
\begin{algorithmic}[1]
    \State $\mathbf{u}_i\leftarrow\mathrm{Random}([0,1])$
    \State $\mathbf{u}_i(-1)\leftarrow\mathbf{u}_i(0)$
    \State $t\leftarrow 1$
    \While $~t<T_{\mathrm{max}}$
    \State $\mathbf{u}_i(t)\leftarrow 2\mathbf{u}_i(t-1)-\mathbf{u}_i(t-2)-c^2\sum_{j\in\mathcal{N}(i)}\mathbf{L}_{ij}\mathbf{u}_j(t-1)$
    \State $t\leftarrow t+1$
    \EndWhile
    \State \text{Create the matrices $\mathbf{X}_i, \mathbf{Y}_i\in\mathbb{R}^{K\times M}$ defined by Eqns.~\ref{eqn:X} and \ref{eqn:Y} at node $i$, }\\
    \text{using $\mathbf{u}_i(0),\mathbf{u}_i(1),\dots,\mathbf{u}_i(T_{\mathrm{max}}-1)$, where $K+M=T_{\mathrm{max}}$.}
    \State \text{$\mathbf{v}_i\leftarrow a_i$ from DMD$(\mathbf{X}_i,\mathbf{Y}_i)$ by algorithm~\ref{alg:dmd}}
    \For{$j\in\{1,\dots,k\}$}
    \If{$\mathbf{v}_i^{(j)}>0$}
    \State $\Gamma_j\leftarrow 1$
    \Else 
    \State $\Gamma_j\leftarrow 0$
    \EndIf
    \EndFor
    \State $\text{Cluster number} \leftarrow \sum_{j=1}^k \Gamma_j 2^{j-1}$
\end{algorithmic}
\end{algorithm}

\section{Quantum-Analog Hybrid computation of spectral clustering}
We now adapt our classical algorithm such that various steps are executed on quantum or analog platforms. We now describe the details of executing the various steps on the two platforms along with associated benefits.

\subsection{ Algorithm step 1: wave evolution on graphs}
\label{sec:wave}
Conventional computers suffer from various computational challenges and inherent truncation errors in the simulation process when performing dynamics simulations. In particular, matrix multiplication, central to these computations, has a computational complexity of $\mathcal{O}(N^{2.371866})$ on digital computers \cite{Williams2023NewBF}. The algorithm with this complexity also has reduced numerical stability, compared to the na\"ive $\mathcal{O}(N^3)$ algorithm \cite{Williams2023NewBF}. On the contrary, quantum and analog computers demonstrate inherent prowess in this domain, providing an exponential advantage for simulating dynamical systems compared to their digital counterparts \cite{surana2024efficient,Luca2022optimal,gnanasekaran2023efficient}. For the specific case of wave dynamics, we can choose to simulate the equations in their continuous form on a quantum computer or after discretization on an analog device. Note that, continuous simulations on quantum computers avoids the truncation errors met in discrete simulations on analog or digital computers. We now provide details on both computing options.

\paragraph{Quantum-acceleration of wave dynamics evolution.} The wave equation on a graph~\cite{belkin2001laplacian} is equivalent to equation \ref{eqn:wave}, given as,
\begin{align}
    \frac{\mathrm{d}^2 \mathbf{u}}{\mathrm{d} t^2} = -c^2\mathbf{Lu}
\end{align}
where $\mathbf{L}$ is the graph Laplacian \cite{Friedman2004WaveEF}. In its discrete form, the wave equation on a graph is as follows,
\begin{align}
    \mathbf{u}_i(t) = 2\mathbf{u}_i(t-1)-\mathbf{u}_i(t-2)-c^2\sum_{j\in\mathcal{N}(i)}\mathbf{L}_{ij}\mathbf{u}_j(t-1),
\end{align}

where $\mathbf{u}_i$ is the $i$-th element of $\mathbf{u}$. In~\cite{zhu2022dynamic}, the authors simulated wave dynamics with $dt=1$ on a digital computer. As this simulation only involves matrix-vector multiplication, we can replace this $\mathcal{O}(N^2)$ operation on digital computers  with $\mathcal{O}(N)$ on quantum platforms~\cite{sun2022invited, zhang2016quantum} (or $\mathcal{O}(1)$ operation on analog computers).

Using the above approach, simulating the wave equation on a graph, has a time complexity of $\mathcal{O}(D^{5/2}l/\Delta l+TD^2/\Delta l)$, where $T$ is the terminal time of simulation, $D$ is the dimension of the wave, $l$ is the diameter of the region, and $\Delta l$ is the step size \cite{Costa2017QuantumAF}. Consequently, for the setting of graphs, the dimension $D$ is 2 (time and space), the diameter of the region $l$ is the number of nodes $N$, and the step size $a$ is 1, so the complexity is $\mathcal{O}(N+T)$. In practice, we simulate the wave dynamics for around $2N$ steps, leading to an $\mathcal{O}(N)$ computational complexity.
Since $T=2N$ leads to long simulation times for large graphs, the approximation (round-off) errors accumulate for the discrete system of equations. Therefore, an exact continuous simulation is favorable over their discretized counterparts. 

To implement continuous simulations of the wave dynamics that avoid truncation errors, \cite{Costa2017QuantumAF} suggests using quantum computers by exploiting native Hamiltonian simulations or the quantum linear system algorithm (QLSA). 
In our work, we evolve the following Hamiltonian system on quantum platforms,  
\begin{align}
    \frac{\mathrm{d} \mathbf{u}}{\mathrm{d} t} = -i\mathbf{H} \mathbf{u}(t),
\end{align}
for a predetermined period before measurements collapse the wavefunction. In the above equation, $\mathbf{H}$ is the Hermitian static Hamiltonian \cite{Luca2022optimal}. 
We note the original classical algorithm used the random walk version of the graph Laplacian~\cite{sahai2012hearing}. However, for quantum computations, the symmetric form of the graph Laplacian~\cite{von2007tutorial} is more amenable for Hamiltonian embeddings (see remark~\ref{remark:lsym} for more details). Here, $\mathbf{L}_{\text{sym}}$ is the symmetric normalized graph Laplacian~\cite{von2007tutorial} defined as,
\begin{align}
\label{eqn:glap_sym}
\mathbf{L}_{\text{sym}}=\mathbf{I}-\mathbf{D}^{-1/2}\mathbf{W} \mathbf{D}^{1/2},
\end{align}
where the degree matrix $\mathbf{D}=\text{diag}(\sum_{j=1}^n \mathbf{W}_{1j}, \sum_{j=1}^n \mathbf{W}_{2j}, \dots, \sum_{j=1}^n \mathbf{W}_{nj})$. 

Now one can define the symmetric Laplacian using the normalized signed incidence matrix of the graph (denoted by $\mathbf{B}$) as follows, $\mathbf{L}_{\text{sym}}=\mathbf{BB}^T$. To define $\mathbf{B}$ consider a graph with $N$ nodes and $M$ edges, then the $N\times M$ incidence matrix $\boldsymbol{\iota}$ has rows indexed by nodes $v$ and columns indexed by edges $e$. Specifically, the incidence matrix $\boldsymbol{\iota}$ is defined as,
\[\boldsymbol{\iota}_{ij}=\begin{cases}
    -1 & \text{if edge $e_j$ leaves node $v_i$,}\\
    1 & \text{if edge $e_j$ enters node $v_i$,}\\
    0 & \text{otherwise}.
\end{cases}\]
Now, the normalized incidence matrix $\mathbf{B}$ is defined as, 
\begin{equation}
    \mathbf{B} =  \mathbf{D}^{-\frac12} \boldsymbol{\iota}.
\end{equation}



Since we restrain to undirected graphs, for any given edge $e_j$, we randomly pick one node as the source node and the other as the sink node.

Using the formulation in~\cite{Costa2017QuantumAF}, if we now define a Hermitian Hamiltonian in the following block form,
\begin{align}
\label{eqn:hamiltonian}
    \mathbf{H}=c\begin{bmatrix}
        \mathbf{0} & \mathbf{B}\\
        \mathbf{B} & \mathbf{0}
    \end{bmatrix},
\end{align}
the Schr\"{o}dinger Hamiltonian system takes the form,
\begin{align}
\label{eqn:schrod}
    \frac{\mathrm{d}}{\mathrm{d} t}\begin{bmatrix}
        \mathbf{u}_1\\
        \mathbf{u}_2
    \end{bmatrix}=-ic\begin{bmatrix}
        \mathbf{0} & \mathbf{B}\\
        \mathbf{B} & \mathbf{0}
    \end{bmatrix}\begin{bmatrix}
        \mathbf{u}_1\\
        \mathbf{u}_2
    \end{bmatrix},
\end{align}
which implies,
\begin{align}
\label{eqn:schr_ham}
    \frac{\mathrm{d}^2}{\mathrm{d} t^2}\begin{bmatrix}
        \mathbf{u}_1\\
        \mathbf{u}_2
    \end{bmatrix}=-c^2\begin{bmatrix}
        \mathbf{0} & \mathbf{B}\\
        \mathbf{B} & \mathbf{0}
    \end{bmatrix}^2\begin{bmatrix}
        \mathbf{u}_1\\
        \mathbf{u}_2
    \end{bmatrix}
    =-c^2\begin{bmatrix}
        \mathbf{BB}^T & \mathbf{0}\\
        \mathbf{0} & \mathbf{B}^T\mathbf{B}
    \end{bmatrix}\begin{bmatrix}
        \mathbf{u}_1\\
        \mathbf{u}_2
    \end{bmatrix}=-c^2\begin{bmatrix}
        \mathbf{L}_{\text{sym}} & \mathbf{0}\\
        \mathbf{0} & \mathbf{B}^T\mathbf{B}
    \end{bmatrix}\begin{bmatrix}
        \mathbf{u}_1\\
        \mathbf{u}_2
    \end{bmatrix}.
\end{align}
When simulating dynamics of the Schr\"{o}dinger system using Eqn.~\ref{eqn:schr_ham}, the first $N$ terms correspond to the desired wave dynamics. When $\mathbf{u}$ has a non-zero initial derivative condition, the resulting solution tends to increase over time, leading to instability \cite{evans2022partial}. Given that:
\begin{align}
\label{eqn:wave_itr_sym}
    \frac{\mathrm{d}}{\mathrm{d} t}\begin{bmatrix}
        \mathbf{u}_1\\
        \mathbf{u}_2
    \end{bmatrix}=-ic\begin{bmatrix}
        \mathbf{0} & \mathbf{B}\\
        \mathbf{B} & \mathbf{0}
    \end{bmatrix}\begin{bmatrix}
        \mathbf{u}_1\\
        \mathbf{u}_2
    \end{bmatrix}=-ic\begin{bmatrix}
        \mathbf{Bu}_2\\
        \mathbf{Bu}_1
    \end{bmatrix},
\end{align}
the stability criteria $\frac{\mathrm{d}\mathbf{u}_1}{\mathrm{d} t}|_{t=0}=\mathbf{0}$ translates to the initial condition $\mathbf{u}_2(0)=0$. Therefore, all non-zero initial values for $\mathbf{u}_2$ lead to unstable dynamics.

\begin{remark}
\label{remark:lsym}
We again note that since the graph Laplacian form used in~\cite{sahai2012hearing, zhu2022dynamic} is the random-walk variant of the graph Laplacian $\mathbf{L}_{\text{rw}}$, it does not have a Hermitian form. Consequently, the wave dynamics based on $\mathbf{L}_{\text{rw}}$ cannot be evolved using quantum platforms. Consequently, on quantum devices, one is restricted to the symmetric graph Laplacian $\mathbf{L}_{\text{sym}}$ form defined in Eqn.~\ref{eqn:glap_sym}. Moreover, since the simulation is Hamiltonian, it is guaranteed to remain stable on quantum platforms. 
\end{remark}

\begin{remark}
In equation~\ref{eqn:schrod}, the initial state can be generated by starting from a fixed reference quantum state $\psi_{0}$ and transforming it to $|\psi_U \rangle= U|\psi_{0}\rangle$, where $U$ is generated from the uniform Haar measure~\cite{reck1994experimental,tang2022generating}. This result is an natural consequence of the results derived on the initial condition of the wave equation evolution derived in~\cite{sahai2012hearing}. Therefore, state preparation will not be a significant bottleneck in the implementation of our approach.
\end{remark}

\paragraph{Analog computer-based acceleration of wave dynamics.}
When working with the discrete form of the equation, one can exploit the favorable $\mathcal{O}(1)$ scaling of analog platforms for matrix-vector multiplication operation. Unlike the quantum setting, the $\mathbf{L}_{\text{rw}}$ update equation that is used in the classical approach~\cite{sahai2010wave,sahai2012hearing,zhu2022dynamic} can be used on analog platforms. However, in the following lemma, we prove that the wave dynamics with $\mathbf{L}_{\text{sym}}$ remain stable under the same conditions that were derived in~\cite{sahai2012hearing} for $\mathbf{L}_{\text{rw}}$, thereby significantly extending the original results for digital and analog settings.
\begin{lemma}
\label{lemma:stable_symm}
    The wave equation iteration given by Eqn.~\ref{eqn:wave_itr_sym} is stable on any graph as long as the numerical value of the wave speed $c$ satisfies $0<c<\sqrt{2}$ and the iterations have initial conditions that satisfy $\mathbf{u}_1(-1)=\mathbf{u}_1(0)$.
\end{lemma}

\begin{proof}
For the $\mathbf{L}_{\text{sym}}$ case, the  dynamics of $\mathbf{z}(t) =[\mathbf{u}(t);\mathbf{u}(t-1)]$ evolves as,
\begin{align*}
    \mathbf{z}(t) &= \begin{bmatrix}
        \mathbf{u}(t)\\
        \mathbf{u}(t-1)
    \end{bmatrix}=\begin{bmatrix}
        2\mathbf{I}-c^2\mathbf{L}_{\text{sym}} & -\mathbf{I}\\
        \mathbf{I} & \mathbf{0}
    \end{bmatrix}\begin{bmatrix}
        \mathbf{u}(t-1)\\
        \mathbf{u}(t-2)
    \end{bmatrix}=\mathbf{M}\mathbf{z}(t-1) = \mathbf{M}^t \mathbf{z}(0).
\end{align*}
Let matrix $\mathbf{M}$ have eigenvector $\mathbf{m}_j = (\mathbf{m}^{(1)}_j,\mathbf{m}^{(2)}_j)^T$, where $j$ denotes the node number and the superscript splits the compound vector into two equal parts. Now, let the corresponding eigenvalue be $\rho_j$, it then follows that,

\begin{align*}
\mathbf{M}\begin{bmatrix}
    \mathbf{m}^{(1)}_j\\
    \mathbf{m}^{(2)}_j
\end{bmatrix}
&=\begin{bmatrix}
    2\mathbf{I} - c^2 \mathbf{L}_{\text{sym}} & -\mathbf{I}\\
    \mathbf{I} & \mathbf{0}
\end{bmatrix}\begin{bmatrix}
    \mathbf{m}^{(1)}_j\\
    \mathbf{m}^{(2)}_j
\end{bmatrix}=\begin{bmatrix}
    2\mathbf{m}^{(1)}_j-c^2\mathbf{L}_{\text{sym}}\mathbf{m}^{(1)}_j-\mathbf{m}^{(2)}_j\\
    \mathbf{m}^{(1)}_j
\end{bmatrix},\text{and }\;
\end{align*}
\begin{align*}
\mathbf{M}\begin{bmatrix}
    \mathbf{m}^{(1)}_j\\
    \mathbf{m}^{(2)}_j
\end{bmatrix}
=\rho_j \begin{bmatrix}
    \mathbf{m}^{(1)}_j\\
    \mathbf{m}^{(2)}_j
\end{bmatrix}
=\begin{bmatrix}
    \rho_j  \mathbf{m}^{(1)}_j\\
    \rho_j  \mathbf{m}^{(2)}_j
\end{bmatrix}.
\end{align*}

If we let $\mathbf{m}^{(2)}_j=\mathbf{v}_j$ (where $\mathbf{v}_j$ is the eigenvector of $\mathbf{L}_{\text{sym}}$),  the above equation implies that $\mathbf{m}^{(1)}_j=\rho_j \mathbf{m}^{(2)}_j=\rho_j \mathbf{v}_j$. Therefore, the eigenvector $\mathbf{m}_j$ of $\mathbf{M}$ is $\mathbf{m}_j=(\rho_j \mathbf{v}_j, \mathbf{v}_j)^T$. 

Using the definition of eigenvalues and pulling everything to the left hand side of the equation, 
\begin{align*}
    \mathbf{M}\begin{bmatrix}
        \rho_j \mathbf{v}_j\\
        \mathbf{v}_j
    \end{bmatrix} &= \rho_j \begin{bmatrix}
        \rho_j \mathbf{v}_j\\
        \mathbf{v}_j
    \end{bmatrix}\Rightarrow
    \begin{bmatrix}
        (\rho_j(2-\rho_j)\mathbf{I}-\rho_j c^2 \mathbf{L}_{\text{sym}}-\mathbf{I})\mathbf{v}_j\\
        \mathbf{0}
    \end{bmatrix}= \begin{bmatrix}
        \mathbf{0}\\
        \mathbf{0}
    \end{bmatrix}.
\end{align*}

If we now denote the eigenvalue of $\mathbf{L}_{\text{sym}}$ as $\lambda_j$, we get,
\begin{align*}
    \rho_j (2-\rho_j)\mathbf{v}_j - c^2\rho_j \mathbf{L}_{\text{sym}} \mathbf{v}_j - \mathbf{v}_j &= 0 \Rightarrow
    \rho_j^2+(c^2\lambda_j-2)\rho_j+1 = 0,
\end{align*}
so we have
\begin{align*}
    \rho_j &= \frac{-(c^2\lambda_j -2)\pm\sqrt{(c^2\lambda_j -2)^2-4}}{2}
    = \frac{2-c^2\lambda_j}{2}\pm\frac{c}{2}\sqrt{c^2\lambda_j^2-4\lambda_j}
\end{align*}
Therefore, to generate stable dynamics with $\mathbf{M}$, we need $\mathbf{L}_{\text{sym}}$'s eigenvalues $\lambda_j$'s to satisfy $c^2\lambda_j^2-4\lambda_j< 0$. If we now use that the fact that $0\leq \lambda_j\leq 2$ (see~\cite{von2007tutorial} for more details), we get $0\leq c<\sqrt{2}$. Thus, as long as the above conditions are true, the update equations remain stable on analog devices.

\end{proof}

\subsection{Algorithm step 2: Eigenvalue/vector computations}
\label{sec:eigen}
\paragraph{Analog computer-based acceleration of eigenvector computations} Once the wave equation is evolved on either analog or quantum platforms, one has to compute the eigenvalues and eigenvector components for node $l$ using the DMD approach outlined previously. We begin with a reduced SVD of matrix $\mathbf{X}$. Despite the superiority of existing quantum SVD methods ($\mathrm{poly log}$ complexity), they only work on low-rank matrices, which is not the case here \cite{PhysRevA.97.012327}. We instead focus on efficient methods for SVD that exploit matrix-vector multiplication and are, therefore, tailored for analog computers. It's important to note that the eigenvalues of $\mathbf{X}^T\mathbf{X}$, represented as $\sqrt{\gamma_i}$, correspond to the squared singular values of $\mathbf{X}$ \cite{Anthony2006AdvancedLA}. Since $\mathbf{X}^T\mathbf{X}$ is symmetric, it has distinct real eigenvalues, which can be determined using the power method. 

The power method approximates dominant eigenvalues (and their corresponding eigenvectors) of matrices. Especially powerful for large-scale matrices, this method iteratively refines estimates of the dominant eigenvector, capitalizing on the property that repeated multiplication accentuates the contribution of the dominant eigenvalue. The following recurrence relation describes the power method:

\begin{align}
    \mathbf{b}_{k+1} = \frac{\mathbf{X}^T\mathbf{X}\mathbf{b}_k}{\|\mathbf{X}^T\mathbf{X}\mathbf{b}_k\|}.
\end{align}
Using the eigenvector $\mathbf{b}_k$, the estimation for the dominant eigenvalue is the Rayleigh quotient 
\begin{align}
    \gamma_k = \frac{\mathbf{b}_k^T\mathbf{X}^T\mathbf{X}\mathbf{b}_k}{\mathbf{b}_k^T\mathbf{b}_k} \text{\cite{Anthony2006AdvancedLA}}.
\end{align}

To calculate other eigenvalues and eigenvectors using the power method, we need to remove the dominant eigenvalue while preserving the spectrum. Deflation methods such as Wielandt's and Hotelling's deflation provide such tools \cite{Saad2011NumericalMF,NIPS2008_85d8ce59}. In particular, Hotelling's deflation states that

\begin{lemma}
    If $\psi_1\geq\psi_2\geq\dots\geq\psi_N$ are the eigenvalues of matrix $\,\mathcal{B}$, and $\boldsymbol{\nu}^{(1)},\boldsymbol{\nu}^{(2)},\dots,\boldsymbol{\nu}^{(N)}$ are the corresponding eigenvectors, define $\hat{\mathcal{B}}=\mathcal{B}-\boldsymbol{\nu}^{(j)}\boldsymbol{\nu}^{(j)T}\mathcal{B}\boldsymbol{\nu}^{(j)}\boldsymbol{\nu}^{(j)T}$ for some $j\in\{1,2,\dots,N\}$. Then $\hat{\mathcal{B}}$ has eigenvectors $\boldsymbol{\nu}^{(1)},\boldsymbol{\nu}^{(2)},\dots,\boldsymbol{\nu}^{(N)}$ with corresponding eigenvalues $\psi_1,\dots,\psi_{j-1},0,\psi_{j+1},\dots,\psi_N$ \cite{NIPS2008_85d8ce59}.
\end{lemma}

Hotelling's deflation maintains the full spectrum of a matrix, nullifying a chosen eigenvalue while keeping the rest intact. Implementation of Hotelling's deflation is straightforward with just matrix-vector multiplications. Due to the approximation error inherent in the power method, Hotelling's deflation does not \emph{exactly} eliminate the eigenvalues, causing numerical instabilities for large matrices. Schur's deflation method mitigates these challenges and remains stable despite the numerical errors in the eliminated eigenvalues. When the elimination is exact, Schur's deflation reduces to Hotelling's deflation \cite{NIPS2008_85d8ce59}. The power method struggles to determine smaller eigenvalues due to numerical issues. Since we exploit the reduced SVD (where we disregard exceedingly small singular values), these numerical issues do not affect the efficiency of our algorithm \cite{Anthony2006AdvancedLA}.

Let us define the SVD of $\mathbf{X}$ as $\mathbf{X}=\mathbf{U}\mathbf{\Sigma}\mathbf{V}^T$. The eigendecomposition with the power method and Hotelling's deflation provides us the singular values of $\mathbf{X}$ (given by the diagonal of $\mathbf{\Sigma}$) as well as the eigenvectors of $\mathbf{X}^T\mathbf{X}$ (given by the columns of $\mathbf{V}$). When performing a reduced SVD, we keep only $M\ll N$ singular values and columns of $\mathbf{V}$. We can calculate the reduced $\mathbf{U}$ using the pseudoinverse of $\mathbf{V}^T_{\mathrm{red}}$, given by
\begin{align}
    \mathbf{U}_{\mathrm{red}}\approx \mathbf{X}\mathbf{V}^{T+}_{\mathrm{red}}\mathbf{\Sigma}^{-1}_{\mathrm{red}},
\end{align}
where $\mathbf{V}^{T+}_{\mathrm{red}}$ is the pseudoinverse of $\mathbf{V}^T_{\mathrm{red}}$. Both quantum and analog computers provide efficient methods for computing the pseudoinverse \cite{sun2022invited,Wiebe2012QuantumD} and in our approach either can be used. 

The subsequent step within DMD involves calculating the eigenvalues of $\Tilde{\mathbf{A}} = \mathbf{U}^*\mathbf{YV\Sigma}^{-1}$. The power method is inapplicable in the setting of general unsymmetric matrices such as $\Tilde{\mathbf{A}}$. However, given that $\Tilde{\mathbf{A}}$ is of dimensions $\mathbb{R}^{M\times M}$ and $M\ll N$,  eigendecomposition techniques are not impacted by large $N$ values. Eigendecomposition on analog computers has a scaling of $\mathcal{O}(M\log M)$~\cite{sun2022invited, 8914709, pan2024blockamcscalableinmemoryanalog}. Although quantum algorithms provide the same scaling, existing quantum algorithms are restricted to matrices with specific properties, such as a Hermitian structure and diagonalizability~\cite{Nghiem_2023, 10.1145/3527845, qpca}. Consequently, in our approach, we exploit analog platforms for these computations due to their flexibility. As quantum algorithms develop further, it is possible that future variants of this algorithm will rely of quantum devices.

\subsection{Algorithm step 3: Solving linear systems}
\label{sec:linear}
\paragraph{Quantum-accelerated solutions of linear systems.} As mentioned previously, the final step within the step of DMD computations is to solve a system of linear equations of the form $\hat{\mathbf{\Phi}}\hat{\mathbf{a}}=\mathbf{x}(0)$. Here, $\hat{\mathbf{a}}$ provides the coefficients whose signs are used to assign the nodes to corresponding clusters. Quantum computers excel at this task, with time complexity of $\mathcal{O}(\text{poly}(\log N, \kappa))$, where $N$ and $\kappa$ are the dimension and condition number of $\hat{\mathbf{\Phi}}$. In contrast, classical algorithms have a scaling of $\mathcal{O}(N\sqrt{\kappa})$ (we refer the reader to~\cite{harrow2009quantum} for further details on solving linear systems on quantum platforms). 

\subsection{Summary of steps}
In summary, the algorithm steps and their associated scaling on the various computing platforms is shown in table~\ref{table:steps}. This table represents the best known scaling of the current generation of algorithms. As new algorithms for quantum and analog platforms are developed, the benefits and drawbacks associated with executing various steps of spectral clustering on these platforms are expected to evolve.
\begin{table}[hbt!]
\begin{tabular}{l|lll}
                             & Digital                           & Analog                            & Quantum                      \\ \hline
Simulate Wave Dynamics       & $\mathcal{O}(N^2 N_t)$             & $\mathcal{O}(N_t)$        & $\mathcal{O}(N)$ \\
Reduced SVD                  & $\mathcal{O}(M^2 N + M^3)$ & $\mathcal{O}(KM)$ & $\mathcal{O}(\mathrm{poly }\log N)$                          \\
Eigendecomposition           & $\mathcal{O}(M^3)$                & $\mathcal{O}(M\log M)$                              & $\mathcal{O}(\mathrm{poly }\log N)$                         \\
Matrix-vector multiplication & $\mathcal{O}(N^2)$                & $\mathcal{O}(1)$                  & $\mathcal{O}(N)$                          \\
Solving linear system        & $\mathcal{O}(N^3)$                & $\mathcal{O}(\log N)$             & $\mathcal{O}(\log N)$       
\end{tabular}
\caption{A summary of the computational complexities of various components of the clustering algorithm. $H$ is the Hermitian static Hamiltonian, $N$ is the number of nodes, $N_t$ is the number of time steps, $K$ is the number of iterations taken for the power method, $M$ is the number of components left in reduced SVD, $T$ is terminal time \cite{harrow2009quantum, sun2022invited, li2019tutorialsvd}. The complexity of quantum matrix-vector multiplication listed above is adapted from quantum matrix multiplication algorithms, which have a complexity of $\mathcal{O}(N^2)$ ~\cite{shao2018quantumalgorithmsmatrixmultiplication}. \cite{Mizuno_2024} proposes a quantum version of the DMD algorithm to achieve an $\mathcal{O}(\mathrm{poly }\log N)$ complexity for eigendecomposition.}
\label{table:steps}
\end{table}

\begin{figure*}[hbt!]
    \centering
    \includegraphics[width=0.6\textwidth, trim=25pt 10pt 80pt 20pt]{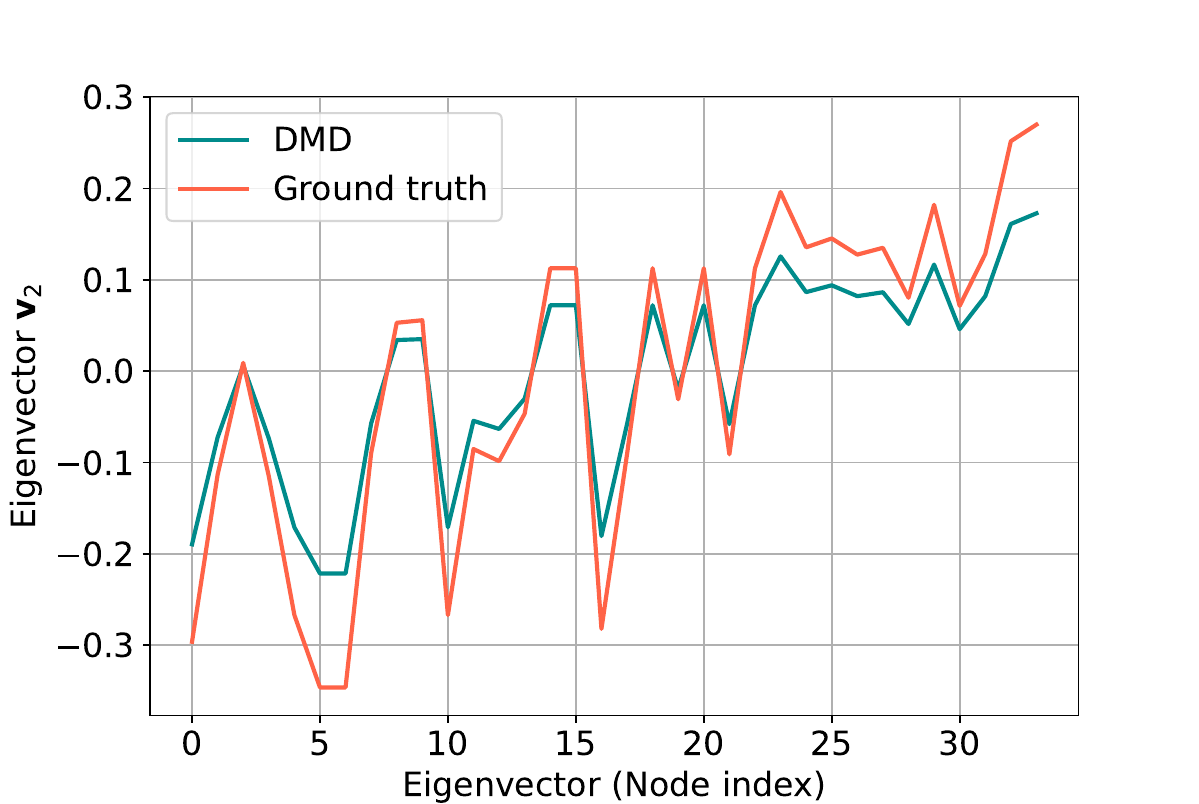}
    \caption{The estimated and the ground truth eigenvector of the adjacency matrix of the Karate club network. Eigenvector (Node index) }
    \label{fig:karate}
\end{figure*}

\section{Results}
\label{sec:result}
We implemented our proposed algorithm using the IBM Qiskit package and performed numerical experiments to verify that its results match those of classical algorithms. Our computations are performed on quantum emulators since the current generation of quantum computers do not have enough qubits to handle meaningful examples. 

First, as shown in Figure \ref{fig:karate}, we experiment with a benchmark called Zachary's karate club graph \cite{Girvan_2002}. Zachary's Karate Club graph is a widely studied social network representing the relationships amongst a karate club's 34 members at a US university in the 1970s. In the graph, each node represents a member of the karate club, and each edge indicates a tie between two members outside of the club activities. The network is undirected and unweighted, reflecting the mutual nature of social relationships. We simulated the wave dynamics using Qiskit for the time interval $t=[0,99]$, measuring the dynamics at integer times \cite{qiskit_dynamics_2023}. As expected, the eigenvector estimated by the proposed method is a constant factor of the ground truth (see Fig.~\ref{fig:karate}). Since the ground truth and computed eigenvectors share consistent signs for assigning nodes to clusters, the estimated clusters exactly match.

Second, we test the performance of our method on the Twitter interaction network for the US Congress, which represents the Twitter interaction network of both the House of Representatives and Senate that formed the 117\textsuperscript{th} the United States Congress \cite{fink2023twitter, fink2023centrality}. As shown in Figure \ref{fig:twitter}, our method correctly classifies 467 out of 475 nodes, thereby achieving an accuracy of $98.32\%$. This discrepancy arises due to the numerical error of simulating quantum dynamics using the Qiskit dynamics package. We also test the proposed method on the social circle graph on Facebook \cite{NIPS2012_7a614fd0}. In particular, we take $200$ nodes from the dataset. As shown in figure~\ref{fig:facebook}, our method correctly clusters $196$ out of the $200$ nodes, achieving an accuracy of $98\%$. We expect $100\%$ accuracy of cluster assignments on an error-corrected (fault tolerant) quantum platform. We note that the error of computing the eigenvector components (from the Schr\"{o}dinger dynamics data) was below the threshold required to perform accurate cluster assignments and did not contribute to the error rates.

To numerically verify the stability of the wave equation simulated with the $\mathbf{L}_\text{sym}$ (proved in lemma~\ref{lemma:stable_symm}), we plot the wave dynamics used for clustering on the Twitter dataset in Fig.~\ref{fig:qiskitWave}. Here, it is easy to see that the dynamics at the second node remains bounded for all time, as should be expected in the evolution of Hamiltonian systems. Similar evolution plots can be generated at all nodes in the graph. Note that the Qiskit dynamics are not expected to match the discrete equation evolution since the governing equations for both the cases are different.   

\begin{figure*}[hbt!]
    \centering
    \includegraphics[width=0.6\textwidth, trim=50pt 20pt 30pt 10pt]{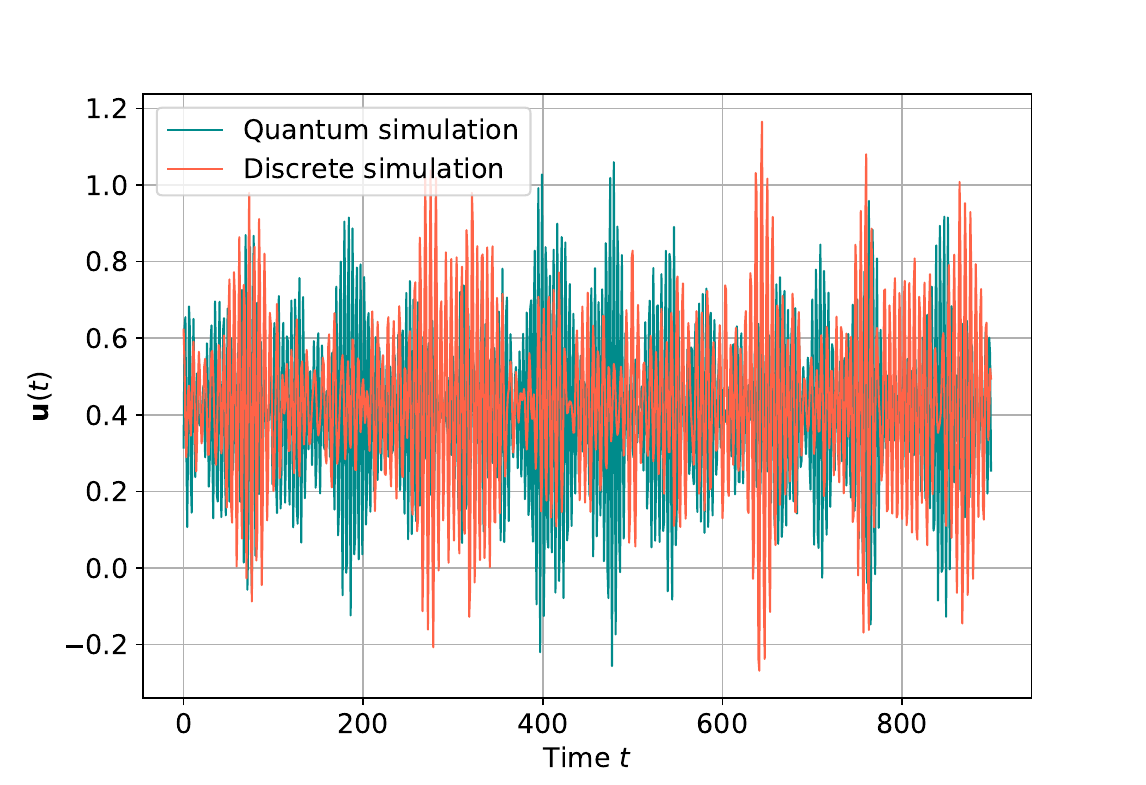}
    \caption{The wave dynamics of the second node on the Twitter dataset simulated with Qiskit dynamics verify the stability of quantum-simulated dynamics.}
    \label{fig:qiskitWave}
\end{figure*}

\begin{figure*}[hbt!]
    \centering
    \begin{subfigure}[t]{0.45\textwidth}
    \centering
    \includegraphics[width=\textwidth, trim=100pt 40pt 10pt 10pt]{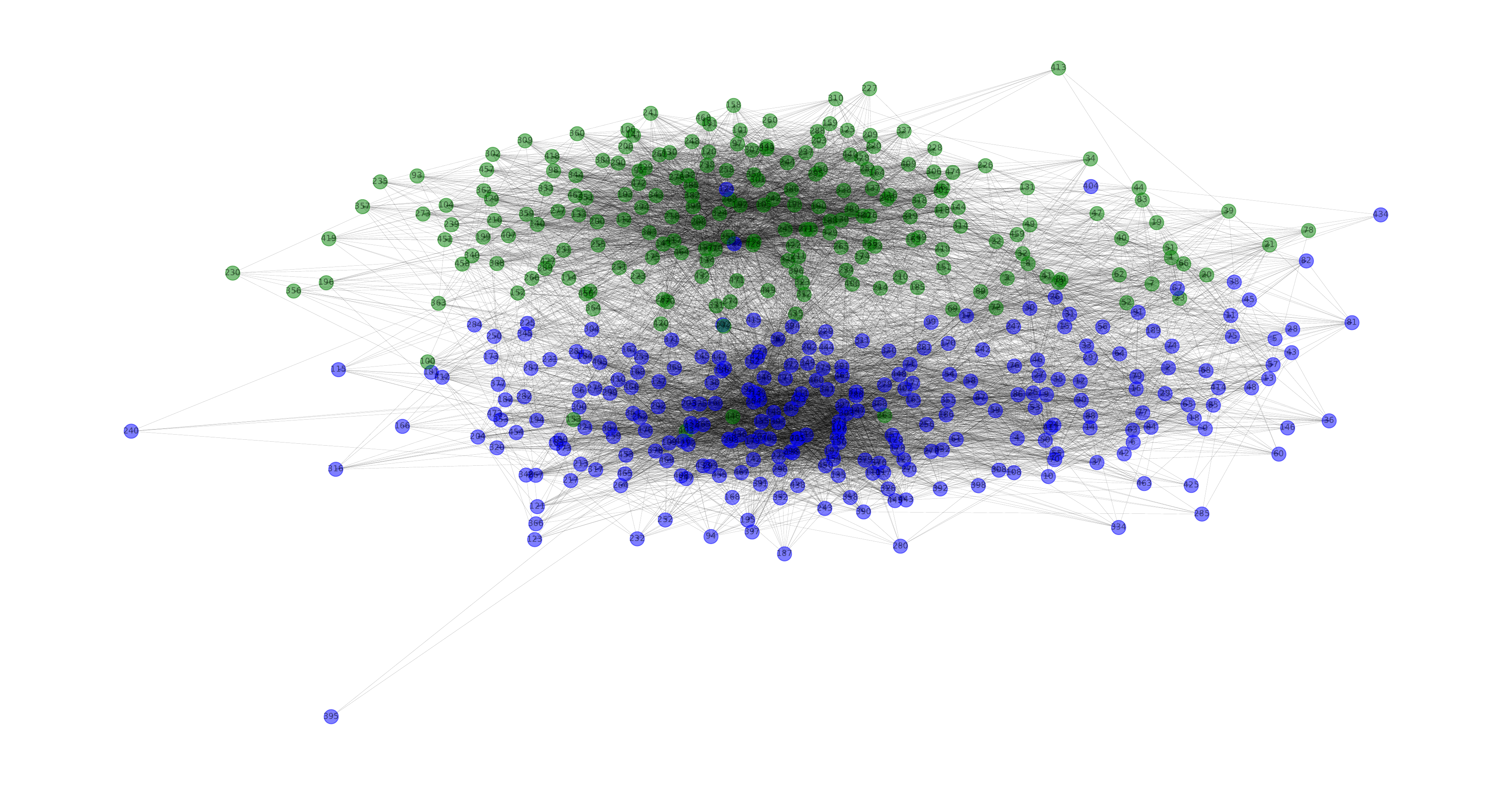}
    \caption{The clusters calculated according to the eigenvector estimated with the proposed method.}
    \label{fig:twitter_dmd}
    \end{subfigure}
    \begin{subfigure}[t]{0.45\textwidth}
    \centering
    \includegraphics[width=\textwidth, trim=30pt 40pt 80pt 10pt]{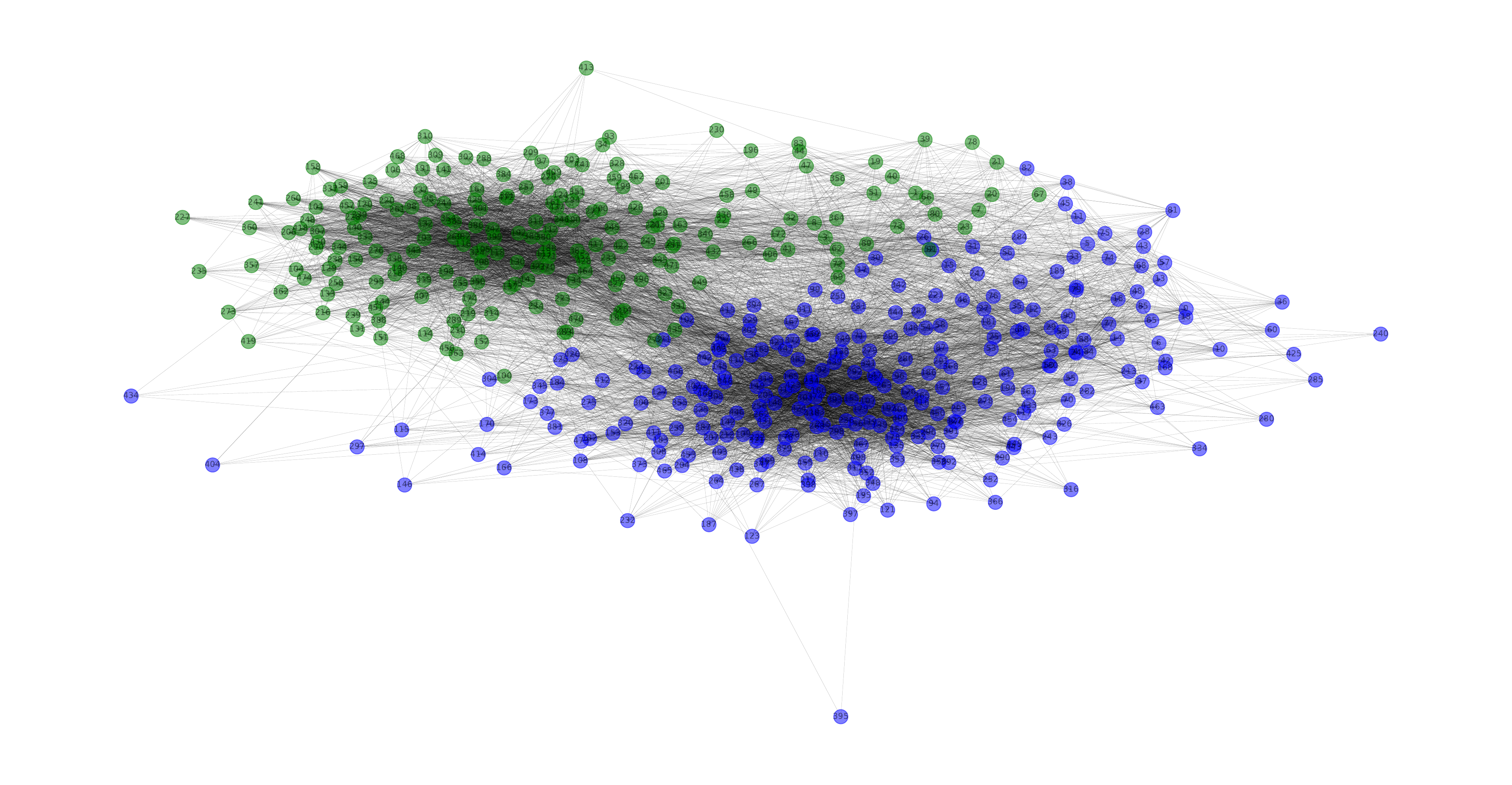}
    \caption{Ground Truth cluster calculated with spectral clustering.}
    \label{fig:twitter_gt}
    \end{subfigure}
    \caption{Clusters of the Twitter interaction network for the US Congress, comparing the ground truth and the proposed method's results.} 
    \label{fig:twitter}
\end{figure*}

\begin{figure*}[hbt!]
    \centering
    \includegraphics[width=1.0\textwidth, trim={0pt 0pt 0pt 0pt}, clip]{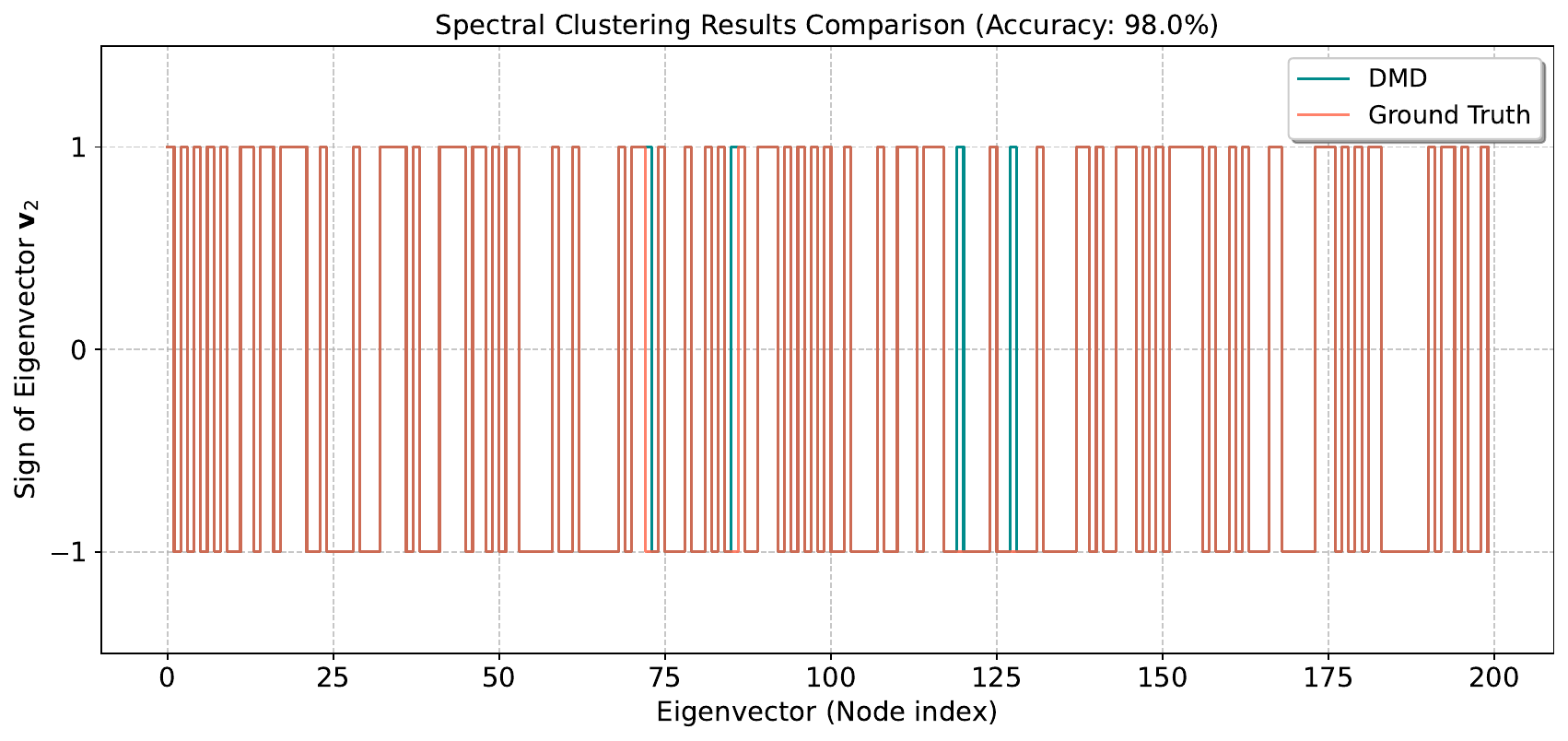}
    \caption{Comparison of the ground truth spectral clustering results and from the proposed quantum-analog framework on $200$ nodes from the Facebook social network graph.}
    \label{fig:facebook}
\end{figure*}

Lastly, we test the proposed method on a synthetic network with eighty nodes and four clusters. The network is generated such that it contains four clusters of nodes, where the nodes within each cluster are closely connected, and the four clusters themselves have weak interaction. In this case, as can be seen in figure~\ref{fig:synthetic}, our method perfectly estimates the clusters with $100\%$ accuracy. 
All results can be reproduced with our codebase at \url{https://github.com/XingziXu/quantum-analog-clustering}.

\begin{figure*}[hbt!]
    \centering
    \begin{subfigure}[t]{0.45\textwidth}
    \centering
    \includegraphics[width=\textwidth, trim=70pt 25pt 10pt 10pt]{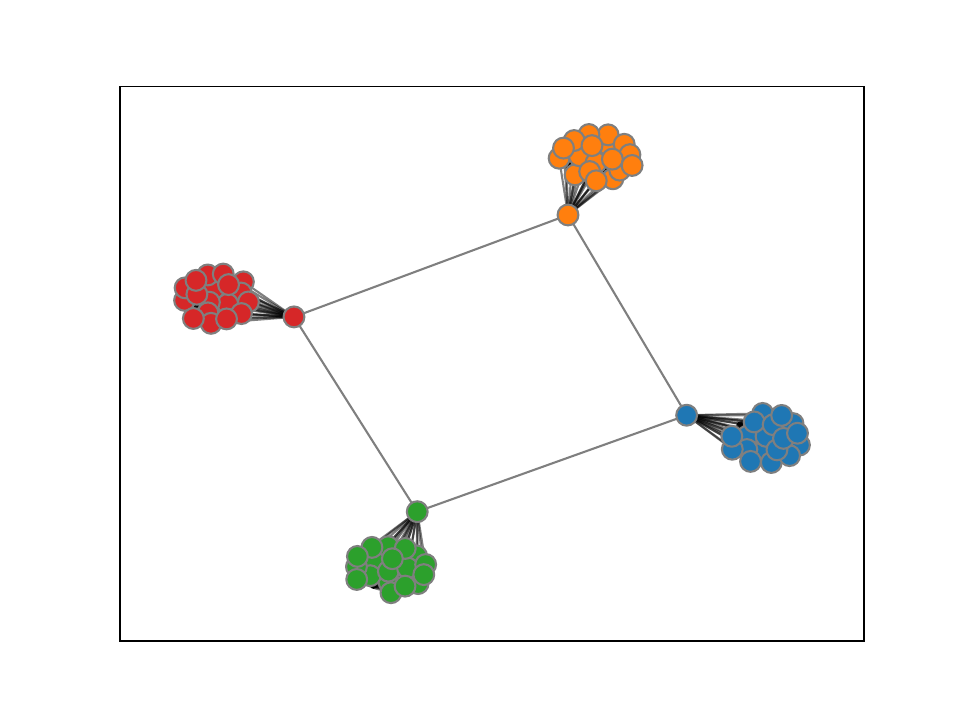}
    \caption{The clusters calculated according to the eigenvector estimated with the proposed method.}
    \label{fig:synthetic_clust}
    \end{subfigure}
    \begin{subfigure}[t]{0.45\textwidth}
    \centering
    \includegraphics[width=\textwidth, trim=30pt 15pt 100pt 30pt]{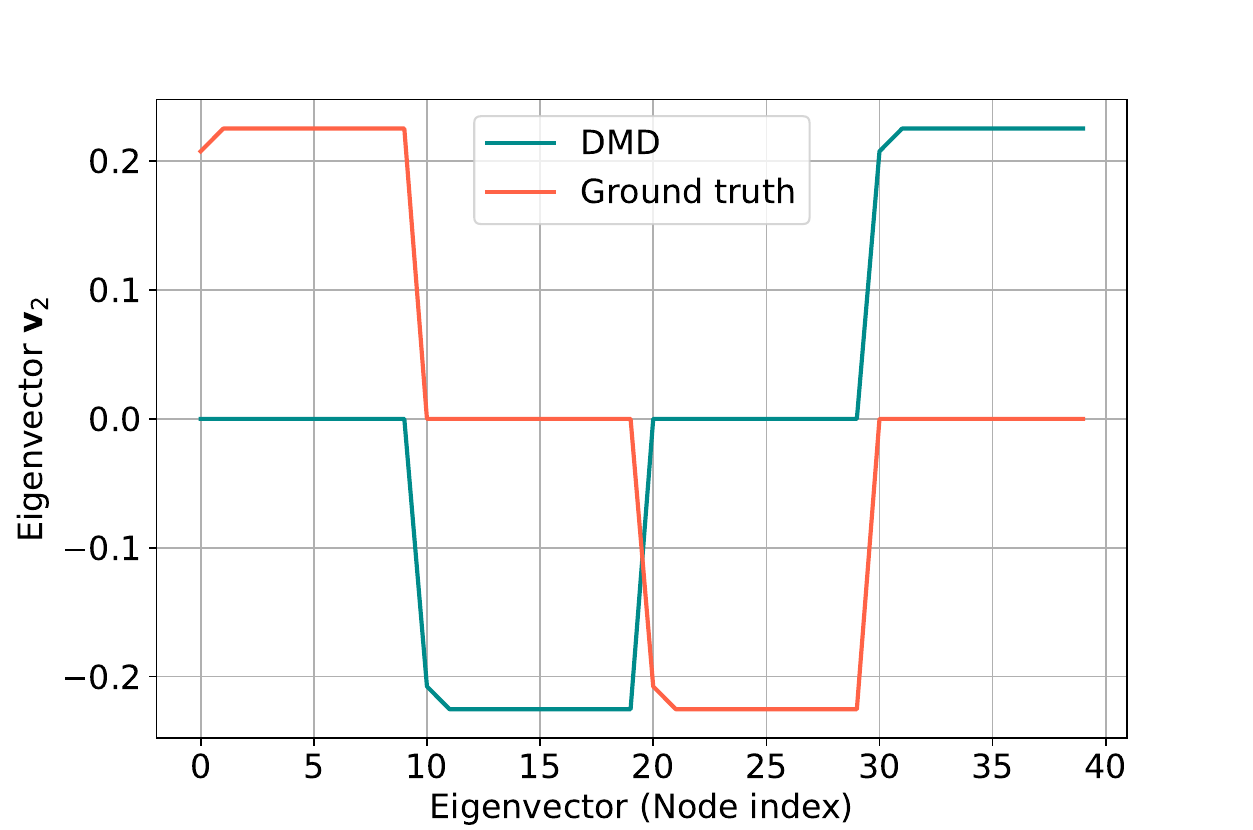}
    \caption{The estimated and the ground truth eigenvector of the adjacency matrix of the network.}
    \label{fig:synthetic_eig}
    \end{subfigure}
    \caption{Estimated clusters of the synthetic network and the comparison of the estimated and ground truth eigenvector.} 
    \label{fig:synthetic}
\end{figure*}

\section{Conclusion}\label{sec:conclusion}
We propose a first-of-a-kind quantum-analog hybrid algorithm to cluster large-scale graphs. The algorithm contains two  components that are sequentially applied. In the first component, we simulate wave dynamics on a graph using analog (discretized dynamics) or quantum (continuous dynamics) computers. Both analog and quantum platforms accelerate the computations by a polynomial factor. However, we expect fault-tolerant quantum computers to provide a significant advantage over analog computers for the task of simulating wave dynamics due to their inherent lack of discretization error, which provides a key advantage of robustness for large graphs. In the second component, we analyze the data from the wave dynamics step using dynamic mode decomposition on analog platforms. In this step, we achieve polynomial speedups over digital computers, arising from the $\mathcal{O}(1)$ complexity of matrix-vector multiplication on analog computers. Overall, we accelerate the existing $\mathcal{O}(N^3)$ complexity on digital computers to $\mathcal{O}(N)$ by using the combination of quantum and analog platforms. We demonstrate our proposed algorithm on diverse benchmark datasets, including the Zachary Karate club example, Twitter interaction networks, Facebook social circles, and random graphs by reproducing the clusters in these graphs with over $98\%$ accuracy. The error rate is due to the non-native simulation of the Schr\"odinger dynamics and expected to be rectified when running these computations on quantum devices. However, in practice, the devices will have an error rate that will influence the accuracy of the approach. Performing sensitivity analysis of the approach concerning error rates in the analog and quantum portions of the computation is subject to future work.

As both of the computing paradigms mature, we expect our method to become increasingly attractive for practical applications. Currently, quantum computers are limited in the number and fidelity of qubits. Similarly, analog devices are still under development and exist only as experimental prototypes. Consequently, we have used classical emulators for both platforms to perform the above experiments. 

We expect that our method will be particularly useful in the analysis of very large graphs that require dynamic updates to the partitioning. For example, the approach will be useful for partitioning of proliferated satellite networks (for task allocation), formation of teams within heterogeneous autonomous systems, analysis of cellular networks, and biological graphs such as metabolic and genetic networks. The primary drawback of the approach is that it requires access to large-scale quantum and analog devices that can be coupled to one another. Moreover, we assume that the computations are fault-tolerant and error-corrected. Such platforms are not expected to be portable in the near future, restricting their deployment on mobile agents.

Given that the development of algorithms that leverage both analog and quantum computing devices is at a nascent stage, we anticipate that several existing and new algorithms will leverage a combination of these paradigms. However, a comprehensive understanding of the properties of problems and algorithms that make them attractive for these settings remains elusive.

More broadly, this work demonstrates that embeddings of discrete problems in continuous spaces enable the construction of efficient algorithms on emerging computing devices. Given deep connections between dynamical systems and combinatorial optimization problems~\cite{sahai2020dynamical}, we expect that novel combinations of digital, analog, and quantum platforms will provide a unique opportunity for developing novel and efficient algorithms that exploit the inherent strength of each platform. As these computing platforms mature, we expect an increase in the development of such algorithms that exploit the dynamic nature of these problems and devices. These embeddings also provide deep insights into the fundamental limitations of algorithm construction and related complexity classes~\cite{sahai2024emergence}. The use of these embeddings for computational complexity analysis of problem classes on emerging comouting platforms remains an underexplored area of research.

\section*{Acknowledgments}
The authors thank SRI International for financial support for the work under Internal Research \& Development (IRAD) project number IRHOME.A.2023.SNAR.0000.

\section*{Author Contributions}

TS and XX conceived and designed the experiments. XX performed the experiments. XX and TS analyzed the data. TS and XX contributed to theoretical analysis. XX and TS wrote the paper. All authors reviewed and approved the final manuscript.

\section*{Funding Statement}

This research was sponsored by SRI International under internal research and development project [project number: IRHOME.A.2023.SNAR.0000]. This research received no external funding.

\section*{Conflict of Interest}

The authors declare no conflict of interest.

\section*{Competing Interest}

The authors declare no competing interest.

\section*{Data Availability}

The data to the Zachary's Karate Club graph can be found at \url{https://networkx.org/documentation/stable/auto_examples/graph/plot_karate_club.html};
The data to the Twitter interaction network for the US Congress can be found at \url{https://snap.stanford.edu/data/congress-twitter.html}.
\bibliographystyle{unsrt}  
\bibliography{references}

\end{document}